\documentclass[10pt]{article}
\usepackage{geometry}                
\usepackage[parfill]{parskip}    
\usepackage{graphicx}
\usepackage{amssymb}
\usepackage{amsmath}
\usepackage{amsthm}
\usepackage{epstopdf}
\usepackage{verbatim}
\usepackage{mathrsfs}
\usepackage{fancyhdr}
\newtheorem{thm}{Theorem}
\newtheorem{ex}{Example}
\newtheorem{rem}{Remark}
\newtheorem{lem}{Lemma}

\newtheorem{defa}{Definition}
\newtheorem{ass}{Assumption}

\newtheorem{prob}{Problem}

\newcommand{\ind}{1\!\!1}

\newcommand{\nn}{\nonumber}

\newcommand{\M}{\mathcal{M}}
\newcommand{\f}{\mathcal{F}}

\newcommand{\e}{\mathrm{E}}

\newcommand{\R}{\mathbb{R}}

\newcommand{\mep}{\mathbb{P}}
\newcommand{\meq}{\mathbb{Q}}

\newcommand{\q}{\mathcal{Q}}

\newcommand{\dmc}{\mathcal{D}}
\newcommand{\cmc}{\mathcal{C}}
\newcommand{\bmc}{\mathcal{B}}
\newcommand{\ymc}{\mathcal{Y}}
\newcommand{\xmc}{\mathcal{X}}
\newcommand{\g}{\mathcal{G}}

\newcommand{\alp}{\begin{equation*}\left\{\begin{array}{lcl}}
\newcommand{\dal}{\end{array}\right.\end{equation*}}
\newcommand{\alpn}{\begin{equation}\left\{\begin{array}{lcl}}
\newcommand{\daln}{\end{array}\right.\end{equation}}
\newcommand{\bq}{\begin{equation*}}
\newcommand{\eq}{\end{equation*}}		
\newcommand{\bqn}{\begin{equation}}
\newcommand{\eqn}{\end{equation}}
\newcommand{\bqq}{\begin{eqnarray*}}
\newcommand{\eqq}{\end{eqnarray*}}
\newcommand{\bqqn}{\begin{eqnarray}}
\newcommand{\eqqn}{\end{eqnarray}}

\DeclareMathOperator*{\sumpport}{\searrow}

\title{Risk- and ambiguity-averse portfolio optimization with quasiconcave utility functionals} 

\author{Sigrid K\"allblad\thanks{CMAP, Ecole Polytechnique, Paris. Email: sigrid.kallblad@cmap.polytechnique.fr. The work was conducted as part of D.Phil. thesis at University of Oxford and was supported by Santander Graduate Scholarship and the Oxford-Man Institute of Quantitative Finance.}}

\begin{document}

\maketitle

\begin{abstract}

Motivated by recent axiomatic developments, we study the risk- and ambiguity-averse investment problem where trading takes place over a fixed finite horizon and terminal payoffs are evaluated according to a criterion defined in terms of a quasiconcave utility functional. We extend to the present setting certain existence and duality results established for the so-called variational preferences by Schied (2007). The results are proven by building on existing results for the classical utility maximization problem.


\end{abstract}

\section{Introduction}

		The optimal investment problem of choosing the best way to allocate an investor's capital is often formulated as the problem of maximizing, over admissible investment strategies, the expected utility of terminal wealth. The formulation relies on the axiomatic foundation developed by von Neumann and Morgenstern \cite{neumann07} and Savage \cite{Savage:1954wo}. In continuous time optimal portfolio selection, the study dates back to the seminal contributions of Merton \cite{merton,merton71}. In order to formulate the expected utility criterion, the agent needs to specify, on the one hand, her preferences via the investment horizon and the utility function and, on the other, her views about the future in providing the probability measure to compute the expectation. The specification of the latter may, however, itself be subject to uncertainty. This is referred to as ambiguity, or Knightian uncertainty in reference to the original contribution of Knight \cite{Knight:21}. It has been brought to prominence via the Ellsberg paradox \cite{ellsberg1961risk}. 	
		
	

	From a decision theoretic point of view, the issue was addressed in the seminal work of Gilboa and Schmeidler \cite{gilboa89}. They formulated axioms on investors' preferences that should account for aversion against both ambiguity and risk. Specifically, within an Anscombe-Aumann model, the axioms of von Neumann and Morgenstern were relaxed in that the axiom of independence was replaced by that of certainty independence. That led to a numerical representation of preferences in terms of a \emph{coherent} monetary utility functional. The robust representation of the latter (cf. \cite{delbaen00} for the case of bounded random variables), then yield the following representation of preferences over objects $X$ within some set $\mathcal{X}$: 
	\bqn
		X\longmapsto\inf_{\meq\in\q}\e^\meq\big[U(X)\big],\label{utmmula}
	\eqn
	for some von Neumann Morgenstern utility function $U$ and set $\q$ of probability measures.

	This motivated the study of the risk- and ambiguity-averse investor who trades over a fixed finite horizon in a continuous-time market model and evaluates terminal wealth according to the criterion \eqref{utmmula}. The investment problem associated with such, so-called, multiple-priors preferences has been well-studied. Stochastic control methods have successfully been applied and explicit solutions obtained for the choice of specific market models and utility functions. Specifically, for stochastic factor models and for non-Markovian models, solutions have been obtained, respectively, in terms of PDE's in \cite{hernandez2,owari09} and BSDEs in \cite{muller05,quenez04} (see \cite{tre} for a full list of references). For general market models and utility functions, we mention in particular the work by Quenez \cite{quenez04} and Schied and Wu \cite{wu}. Relying on the results for the classical utility maximization problem in Kramkov and Schachermayer \cite{kramkov,kramkov03}, the authors in \cite{quenez04} and \cite{wu} established a dual formulation and proved existence of an optimizer (see also \cite{burgert05,wittmuss08} for the case including consumption).


	The axiomatic results of Gilboa and Schmeidler were later generalized in Maccheroni et al. \cite{maccheroni06}, where the independence axiom was further relaxed. This led to a numerical representation in terms of a \emph{concave} monetary utility functional. Combined with the generalisation of the representation of coherent utility functionals to concave ones, for the case of bounded random variables obtained in F\"ollmer and Schied \cite{follmer02} and Fritelli and Rosazza-Gianin \cite{fritelli02}, it implied the numerical representation   
	\bqn
		X\longmapsto\inf_{\meq\in\q}\Big(\e^\meq\big[U(X)\big]+\gamma(\meq)\Big),\label{utmvara}
	\eqn
	for some penalty function $\gamma$. While the multiple-priors setup in \eqref{utmmula} is a worst-case approach, the appearance of $\gamma$ enables the investor to weight the possible market models according to their plausibility. This renders the presentation intuitively appealing.

	The investment-problem associated with these so-called variational preferences has also been studied. Particular attention has been paid to the case when the penalty function is given by the relative entropy of $\meq$ with respect to the reference model. Such criteria were introduced already in the seminal work of Hansen and Sargent \cite{anderson03,hansen01}. For this choice, the problem is most naturally formulated with respect to utility from consumption (or stochastic differential utilities) and the natural tool is the theory of BSDEs. While a systematic study was initiated in \cite{skiadas03}, these results have been considerably extended in a number of articles \cite{bordigoni07,chen02,faidi11,jeanblanc12,lazrak03}. For the case of utility from terminal wealth and general variational criteria, stochastic control methods have been successfully applied for the choice of logarithmic utility. This for stochastic factor models as well as non-Markovian ones; see \cite{hernandez,hernandez07b} and \cite{laeven12}, respectively. General existence and duality results have also been established for the variational preferences. Specifically, by use of similar methods to the ones used in \cite{quenez04,wu}, Schied \cite{schied} generalised these results to the concave case.

	The decision theoretic results in \cite{gilboa89} and \cite{maccheroni06}, were recently yet further extended by developments in Cerreia-Vioglio et al. \cite{cerreia}. Observing that all ambiguity averse preferences are axiomatized by a weakening of the independence axiom (the coordinate independence axiom within the Savage setting), the authors in \cite{cerreia} take this to its extreme by essentially removing this axiom altogether. Thereby, a numerical representation in terms of \emph{quasiconcave} utility functionals emerges. Recent advances also yield robust representations of the latter; besides \cite{cerreiab,cerreia}, see Drapeau and Kupper \cite{drapeau10} and Fritelli and Maggis \cite{fritelli11,fritelli12}. This motivates representing preferences via  
	\bqn
		X\longmapsto\inf_{\meq\in\q}G\Big(\meq,\e^\meq\big[U(X)\big]\Big),\label{quasia}
	\eqn
	for some function $G$ which is jointly quasiconvex, lower semicontinuous in its first argument and non-decreasing and right-continuous in its second. Similarly to the multiple-priors and variational cases (cf. \eqref{utmmula} and \eqref{utmvara}), these advances motivate the study of the associated investment problem. The aim herein is to initiate such a study.

Specifically, within a dominated setup, we consider an investor who trades over a fixed finite horizon in a continuous-time market model, evaluates terminal wealth according to \eqref{quasia} and maximizes this quantity over admissible trading strategies. While the investor's risk-aversion is governed by a standard utility function, the ambiguity preferences are, thus, determined by a quasiconcave utility functional. To the best of our knowledge, this problem has not been studied before. Indeed, although the advances in \cite{cerreia} recently motivated the use of quasiconvex risk measures within the area of portfolio optimization (see \cite{mastrogiacomo13}) and quasiconcave utility functions have been previously studied, the risk- and ambiguity-averse utility maximization problem associated with \eqref{quasia} seems to not have been addressed. We note that the notion is \emph{unifying} in the sense that all ambiguity averse preferences, in particular the multiple-priors, entropic and variational ones, are included as special cases. The class of quasiconcave preferences also includes interesting examples which do not correspond to variational preferences. Among others, the so-called smooth preferences axiomatized in \cite{klibanoff05}, which amount to considering a distribution over possible distributions rather than a worst-case approach.

	Our results extend the ones in \cite{schied} (cf. also \cite{quenez04,wu}) in that we prove existence of an optimal strategy and establish certain duality results for the quasiconcave case. As holds for the classical utility maximization, the study of the problem within the dual domain offers various advantages. Most importantly, in the case of robust preferences the dual problem amounts to a search for an infimum whereas the primal problem features a saddle-point. Similarly to \cite{schied}, we prove our results by building on existing results for the classical utility maximization problem (cf. \cite{kramkov,kramkov03}). However, the quasi-concavity, as opposed to the concavity, of the utility functional, required a slightly different approach to the one used in \cite{schied}. In particular, we obtain alternative proofs of some of the results therein.


	The paper is organized as follows. In Section \ref{model4}, the market model and the investment criterion are specified. The main results are presented in Section \ref{main4} while the proofs and further remarks are given in Section \ref{proof4}. The proof of an auxiliary Lemma is deferred to the appendix.

\section{The market model and the investment criterion}\label{model4}

	We consider a fixed finite horizon $T>0$ and a given filtered probability space $(\Omega,\mathbb{F},(\f_t)_{t\in[0,T]},\mep)$, where $\mep$ is the so-called reference measure. We first precise the investor's preferences. Then, in turn, the market model and the investment problem are specified. 

	The investor's \emph{risk} preferences are quantified via a utility function $U:(0,\infty)\to\R$, which is strictly increasing, strictly concave and satisfies the Inada conditions, 
		\bqn
			\lim_{x\to 0}U'(x)=\infty\quad\textrm{and}\quad\lim_{x\to\infty}U'(x)=0.\label{inadaq}
		\eqn
	As argued in the Introduction, following \cite{cerreia}, the investor's \emph{ambiguity} aversion is specified via a quasiconcave utility functional. The latter is a mapping $\phi:L^\infty\to[-\infty,\infty]$ that satisfies quasiconcavity and monotonicity. Recall that a function is quasiconcave if it is the negative of a quasiconvex function and that a function is quasiconvex if its level-sets are convex. For a function $f:\mathcal{X}\to\R$, that is equivalent to 
			$f\big(\lambda x+(1-\lambda)y\big)\le \max\big\{f(x),f(y)\big\}$, $x,y\in\mathcal{X}$.
	Hence, a quasiconcave utility functional (the negative of a quasiconvex risk measure) satisfies
		\bq
			\phi\big(\lambda X+(1-\lambda) Y)\big)\ge \min\big\{\phi(X),\phi(Y)\big\},\quad\textrm{and}\quad
			\phi(X)\ge \phi(Y), \;\;\textrm{if $X\ge Y$, $\mep$-a.s.}
		\eq
	However, it need satisfy neither cash-invariance nor positive homogeneity. 

	The specification of our criterion relies on the robust representation result (Theorem 3.2 in \cite{drapeau10}; see also Theorem 7 in \cite{cerreia}) stating that every $\sigma(L^\infty,L^1)$-upper semicontinuous quasiconcave utility functional $\phi(X)$ admits the (robust) representation
		\bq
			\phi(X)=\inf_{\meq\in\M_1(\mep)}G\Big(\meq,\e^\meq\big[X\big]\Big), 
		\eq
	for some function $G\in\mathcal{G}$, where the set of functions $\mathcal{G}$ is specified next and $\M_1(\mep)$ is the set of $\sigma$-additive probability measures absolutely continuous with respect to $\mep$. Conversely, for any $G\in\mathcal{G}$, the function $\phi$ defined in this way is a upper semicontinuous quasiconcave utility functional.  \\
	
	\begin{defa}\label{defg}
	Let $\mathcal{G}$ denote the set of functions $G:\M_1(\mep)\times\R\to[-\infty,\infty]$ satisfying the following conditions: 

			\begin{itemize}
				\item[i)]{$G(\meq,\cdot)$ is non-decreasing and right-continuous;}
				\item[ii)]{$G$ is jointly quasiconvex;}
				\item[iii)]{$G^-(\cdot,s)$ is weakly lower semicontinuous, where
					\bq
						G^-(\meq,s)=\sup_{t<s}G(\meq,t),\;\; \meq\in\M_1(\mep).
					\eq}
				\item[iv)]{$G$ has an asymptotic maximum in the sense that 
					\bq
						AM(G):=\lim_{s\to\infty}G(\meq,s)=\lim_{s\to\infty}G(\bar\meq,s),\quad\textrm{for all $\meq,\bar\meq\in\M_1(\mep)$}.
					\eq}
			\end{itemize}
		In particular, $G(\meq,\cdot)$ is quasi-linear and upper semicontinuous.
	\end{defa}

	Combined, the above implies that we consider an investor assessing the utility of terminal payoffs, modelled as random variables on $(\Omega,\f,\mep)$, in terms of a robust utility functional of the form \eqref{quasia} with $G\in\g$. For the specific cases when $G$ corresponds to a coherent and concave utility functional, the criterion reduces, respectively, to the multiple priors and variational preferences studied in \cite{quenez04,wu} and \cite{schied}. In the same way as the study of the latter problems were motivated by the axiomatic results in \cite{gilboa89} and \cite{maccheroni06}, the study of the more general quasiconcave case relies on the recent axiomatic extensions in \cite{cerreia}. 

	Without loss of generality, we may precise \eqref{quasia} by writing
		\bqn
			X\to \inf_{\meq\in\q}G\Big(\meq,\e^\meq\big[U(X)\big]\Big), \label{utmquasi2}
		\eqn
	where, for given $G\in\g$,
		\bq
			\q:=\{\meq\in\M_1(\mep): G(\meq,t)<\infty, \;\textrm{for some $t>0$}\}.
		\eq
	We stress that while $U(X)\in L^0$ for the type of payoffs we are to consider (cf. below), the utility functionals defined with respect to $G\in\g$ are defined only for bounded random variables. Indeed, while the theory of quasiconcave risk measures has been extended to more general spaces, we do not restrict to functions $G\in\g$ for which that holds. In doing so, we follow \cite{schied} who in a similar manner studied the risk- and ambiguity-averse investment problem formulated with respect to penalty functions associated with convex risk measures on $L^\infty$. In order to ensure that \eqref{utmquasi2} is well defined, we let $\e^\meq[F]:=\infty$ if $\e^\meq[F^+]=\e^\meq[F^-]=\infty$. Moreover, we extend the domain of $G(\meq,\cdot)$ to the extended real line, in that we define $G(\meq,-\infty):=-\infty$ and $G(\meq,\infty):=AM(G)$.\\ \label{expoperator}
		
		
	\begin{ex}
	An example of preferences defined in terms of a quasiconcave utility functional are the so-called smooth criteria:
		\bqn
			X\longmapsto \phi^{(-1)}\Big(\int_{\M_1(\mep)}\phi\Big(\e^\meq\left[U(X)\right]\Big)d\mu(\meq)\Big),\label{smoothfirst}
		\eqn
	where $\phi$ is an increasing and concave function modelling the ambiguity aversion. That is, rather than taking the infimum over possible models, the agent considers a distribution $\mu$ over them. Such criteria were axiomatized in Klibanoff et al. \cite{klibanoff05} (see also \cite{strzalecki11}) under axioms stronger than those used in \cite{cerreia}. Hence, (for bounded payoffs) these criteria constitute a particular case of the quasiconcave preferences. The specific form of the associated function $G$ is given in \cite{cerreia} and yields further intuition for the criteria. 
	\end{ex}

	Next, we specify the continuous-time market model and the set of admissible trading strategies. The market model is defined by a $\R^d$-valued price process $S_t$, which is assumed to be a semi-martingale on $(\Omega,\mathbb{F},(\f_t),\mep)$. The price-process $S_t$ is assumed to satisfy the condition of NFLVR with respect to the reference-measure $\mep$. For given initial capital $x>0$, and a predictable $S$-integrable trading strategy $\pi_t$, the associated wealth-process is given by
		\bq
			X_t^\pi=x+\int_0^t\pi_sdS_s,\quad t\ge 0.
		\eq
	We denote a trading strategy admissible if $X_t^\pi\ge 0$, $\mep$-a.s. for $t\in [0,T]$ and denote the associated set of wealth-processes by $\mathcal{X}(x)$. 
	
	In consequence, we aim at studying the following investment problem. \\
	
	\begin{prob}\label{problem4}
		For given $G\in\g$, we consider the risk- and ambiguity-averse investment problem of maximizing the functional \eqref{utmquasi2} over admissible terminal payoffs $X_T^\pi$, $X\in\xmc(x)$. The associated value-function $u:\R_+\to\R$ is defined by
	\bqn
		u(x):=\sup_{X\in\xmc(x)}\inf_{\meq\in\q}G\Big(\meq, \e^\meq\big[U\big(X_T^\pi\big)\big]\Big).\label{u}
	\eqn
	\end{prob} 
	
	We stress that the set of admissible strategies $\xmc$ is defined with respect to $\mep$, the role of which is to specify the null-sets rather than representing the most likely model. In particular, we consider a dominated setting in which all measures $\meq\in\q$ are absolutely continuous with respect to the reference measure (for ambiguity averse portfolio optimization with mutually singular measures, see \cite{nutz13} and the references therein). Also note that while $u(x)$ is non-decreasing, it need be neither concave nor continuous (cf. Remark \ref{tcrm}).

	We conclude this section by defining an auxiliary optimization problem which will be crucial for the analysis that follows. For $\meq\ll\mep$, let
		\bqn
			u_\meq(x):=\sup_{X\in\xmc(x)}\e^\meq\big[U(X_T)\big],\label{primal}
		\eqn
	where $\xmc(x)$ is as defined above. We also introduce the (dual) auxiliary problem 
		\bqn
			v_\meq(y):=\inf_{Y\in\ymc(y)}\e\big[ZV(Y_T/Z)\big],\label{dual}
		\eqn
	where $Z=\frac{\mathrm{d}\meq}{\mathrm{d}\mep}$, $V(y)=\sup_{x\ge 0} (U(x)-xy)$ and $\mathcal{Y}(y)$ is the set of all positive $\mep$-supermartingales such that $Y_0=y$ and $XY$ is a $\mep$-supermartingale for all $X\in\mathcal{X}(1)$. Without further notice, these functions will also be denoted by $u_Z(x)$ and $u_Z(y)$, respectively. While the objective functions in \eqref{primal} and \eqref{dual} are defined with respect to the measure $\meq$, the set of admissible strategies and dual objects are defined with respect to the reference measure $\mep$. Hence, while it holds for $\meq\sim\mep$, that these auxiliary problems are, respectively, the standard investment problem and its dual counterpart for the market defined with respect to the measure $\meq$, this need not be the case for $\meq\ll\mep$. In particular, $S_t$ need not satisfy the condition of NFLVR with respect to $\meq$ (see further discussion in Section \ref{aux} and Remarks \ref{normaldual1} and \ref{normaldual} below).

	\section{The main results}\label{main4}
	
	In this section, we present the main results. The proofs are given in Section \ref{proof4}. The following standing assumption is imposed throughout. \\
		
	\begin{ass}\label{standing}
		There exists $\meq_0,\meq_1\in\q$ such that 
		\bqn
			u_{\meq_0}(x_0)<\infty,\textrm{ for some $x_0>0$,}\label{ass0}
		\eqn
		and 
		\bqn
			G\big(\meq_1,u_{\meq_1}(x)\big)<\infty,\textrm{ for all $x>0$.}\label{ass1}
		\eqn
	\end{ass}
	
	Note that assumptions \eqref{ass0} and \eqref{ass1} are disjoint in the sense that none of them implies the other (although, they might both be satisfied for the same $\meq\in\q$). Also note that \eqref{ass0} implies $u_{\meq_0}(x)<\infty$ for all $x>0$ and, furthermore, that Assumption \ref{standing} yields
		\bqn
			u(x)=\sup_{g\in\xmc(x)}\inf_{\meq\in\q}G\big(\meq,\e^\meq\big[U\left(f\right)\big]\big)
			\le \inf_{\meq\in\q}G\big(\meq,u_\meq(x)\big)\;<\;\infty,\quad x>0.\label{finiteness}
		\eqn

	First we establish existence of a solution to Problem \ref{problem4}. The question of uniqueness is discussed below. \\

		\begin{thm}\label{existence}
		Let $G\in\g$ such that Assumption \ref{standing} holds and assume either that
			\bqn
				v_\meq(y)<\infty,\;\; y>0, \label{aspt}
			\eqn
		for each $\meq\in\q$, or that $G(\cdot,t)$, $t\in\R$, is convex and \eqref{aspt} holds for each $\meq\in\q_e$. Then, there exists an optimal terminal payoff $\hat X\in\xmc(x)$ for which the supremum in \eqref{u} is attained. 
		\end{thm}

		We make two general remarks on the assumptions of the above result. \\
		
		\begin{rem}
	That $G(\cdot,t)$, $t\in\R$, is convex is a necessary but not a sufficient condition for the associated utility functional to be concave. This class of preferences therefore includes but is not limited to the variational ones. Nevertheless, the very same methods as introduced in \cite{schied} can be used to establish existence for this case, cf. the discussion in Section \ref{aux}, and in consequence the weaker condition that \eqref{aspt} only holds for $\meq\in\q_e$ suffices. \\
		\end{rem}

		\begin{rem} The assumption that \eqref{aspt} holds for $\meq\in\q$ implies (cf. Lemma \ref{key} below) that $\q=\q^f$, where 
		\bq
			\q^f:=\{\meq\in\q:u_\meq(x)<\infty\}.
		\eq
	In fact, a sufficient condition for Theorem \ref{existence} to hold, is that $\q=\q^f$ and $AE_{+\infty}(U)<1$. Indeed, similarly to the standard case (cf. Note 2 in \cite{kramkov03}), this yields the finiteness of $v_\meq(y)$, $\meq\in\q$. Yet another sufficient condition is that $v_\meq(y)<\infty$, $\meq\in\q_e$, and that for each $\meq_0\in\q$, there exists $\meq_1\in\q_e$ such that
		\bqn
			\e^{\meq_0}[U(g)]\le\e^{\meq_1}[U(g)]<\infty,\;\textrm{ for all }\;g\in\cmc(x).\label{limitasapm}
		\eqn
	Indeed, \eqref{limitasapm} implies that $u_{\meq_0}(y)\le u_{\meq_1}(y)<\infty$, where the last inequality follows from the assumption. Hence, the conjugacy relations of Lemma \ref{key} below holds with respect to $\meq_0$. This yields, 
				\bq
					v_{\meq_0}(y)=\sup_{x>0}\Big(u_{\meq_0}(y)-xy\Big)
					\le \sup_{x>0}\Big(u_{\meq_1}(y)-xy\Big)
					=v_{\meq_1}(y)<\infty.
				\eq
	In consequence, $v_{\meq_0}(y)<\infty$ for all $\meq_0\in\q$ and the assumptions of Theorem \ref{existence} hold.
		\end{rem}

	We stress that the assumption that $v_\meq(y)<\infty$, $\meq\in\q$, is rather natural. First, under some additional assumptions, it is established below (cf. Theorem \ref{interchange}) that the infimum and supremum in \eqref{u} can be interchanged. In effect, the set $\q$ can be replaced by the set $\q^f$ without affecting the indirect utility $u(x)$. More importantly, this assumption implies that the auxiliary investment problem, itself, is solvable for each individual measure $\meq\in\q$ (cf. Lemma \ref{key} below). As pointed out above, the auxiliary problem is not a standard one. Nevertheless this might be understood as an arbitrage condition put on each individual model (cf. \cite{kardaras} where it is shown that the classical utility maximization problem admits a solution if and only if the market satisfies the no arbitrage condition NA$_1$).\label{kardaras} The assumption therefore relates to the interpretation of the ambiguity averse criterion. On the one hand, the criterion \eqref{utmquasi2} emerges due to the axioms posed on the preferences; this via the robust representation of quasiconcave utility functionals. This motivation - per se - does not imply that the measures $\meq\in\q$ satisfy any market related conditions. However, it is common in the literature to also motivate ambiguity averse criteria by the fact that they, effectively, amount to taking the expectation with respect to different possible market models $\meq$ weighted according to their plausibility. While the weighting is determined by the penalty function $\gamma$ for the variational case, the function $G$ allows for more flexibility in the quasiconcave case. Given the latter interpretation, it is natural to assume that each market model $\meq\in\q$, itself constitutes a sensible market model which excludes arbitrage. This is exactly what \eqref{aspt} guarantees. The important question of how to ensure that this condition holds for specific choices of $G$ remains to be addressed.

	Under the assumptions of Theorem \ref{existence}, it is not clear whether the infimum and supremum in \eqref{u} can be interchanged (neither whether the infimum is attained). Hence, there need \emph{not} exist an auxiliary problem (cf. \eqref{primal})\label{equivaux} producing the same investment behaviour as the ambiguity averse criterion \eqref{utmquasi2}. This is, however, the case under some additional assumptions; we will come back to this (see page \pageref{auxreplace}).

	Next, we turn to the study of the dual version of the risk and ambiguity averse investment problem \eqref{utmquasi2}. Specifically, we establish relations between the primal and dual problems and their respective solutions. As holds for the variational case (cf. \cite{schied}), the study of the dual problem is, in fact, not needed for proving existence of a solution to the primal problem. Indeed, the proof of Theorem \ref{existence} relies on properties of the auxiliary problem \eqref{primal} and, thus, rather on the study of the dual version thereof. Nevertheless, the study of its dual counterpart is of interest as it gives a further understanding of the problem and of the optimal strategy. Similarly to the standard case, the study of the dual problem rather than the primal one offers various advantages. This is particularly evident for ambiguity averse preferences where the dual problem amounts to the search for a pure infimum, whereas the primal problem features a saddle-point. In particular, most articles providing explicit solutions for specific choices of utility and penalty functions, notably focus on the dual rather than the primal problem. Hence, there is reason to believe that the dual formulation is helpful in obtaining more explicit results also for the quasiconcave case. The results also enable us to draw important conclusions about the optimal strategy. Specifically, regarding the existence of an equivalent auxiliary problem and uniqueness of the optimal strategy (cf. the discussion after Theorem \ref{relation}). 
		
	We impose the following additional assumption.\\
		
		\begin{ass}\label{dualass}
			The function $G$ is jointly lower semicontinuous and the level sets $\q_t(c)$ are relatively weakly compact, where
				\bq
					\q_t(c):=\big\{\meq\in\q:G(\meq,t)\le c\big\},\quad t\in\R, c\ge 0.
				\eq 
			Furthermore, either $U:\R_+\to\R_+$ or $G(\meq,\cdot)$, $\meq\in\q$, is concave. 
		\end{ass}

		Note that since $\meq\to G(\meq,t)$ is weakly lower semicontinuous, $\q_t(c)$ is weakly closed and, thus, due to Dunford-Pettis theorem, relative weak compactness is equivalent to uniform integrability. Moreover, recall that for functions $G\in\g$, $G^-(\cdot,t)$, $t\in\R$, is lower semicontinuous and $G(\meq,\cdot)$, $\meq\in\q$, is upper semicontinuous. Hence, the assumption that $G$ is jointly lower semicontinuous is, in fact, equivalent to the assumption that $G(\meq,\cdot)$ is continuous. Both these properties are closely related to the fact that one considers (evenly) quasiconcave utility functionals which are continuous not only from above but also from below and, thus, weakly continuous (as opposed to weakly upper semicontinuous). We refer to \cite{cerreiab,bok,fritelli11,fritelli12} for further details but note that this class of quasiconcave utility functionals is natural within the present context. Indeed, in \cite{cerreia} the axioms were formulated (including a continuity axiom) so as to render a numerical representation in terms of such weakly continuous utility functionals.

		While this (natural) additional continuity assumption is crucial for the duality results, the restriction to positive utility functions is needed for technical reasons. The possibility of relaxing this assumption is left for future study. Note, however, that the results in, among others, \cite{quenez04} are established only for positive utility functions. The restriction to $G(\meq,\cdot)$ concave yields, essentially, the variational case treated in \cite{schied}. It is included as a specific case in order to compare our methods and relate our results to the ones therein. \\

	\begin{thm}\label{interchange}
		Let $G\in\g$ such that Assumptions \ref{standing} and \ref{dualass} hold. Then,
			\begin{itemize}
				\item[i)]{the robust value-function satisfies
					\bqn
					u(x):=\sup_{X\in\xmc(x)}\inf_{\meq\in\q}G\Big(\meq,\e^\meq\big[U(X_T)\big]\Big)
					=\inf_{\meq\in\q}\sup_{X\in\xmc(x)}G\Big(\meq,\e^\meq\big[U(X_T)\big]\Big),\label{hoppas}
					\eqn
				and, furthermore, 
					\bqn
						u(x)=\inf_{y>0}v(y;x),\label{udual}
					\eqn
				where $v(y;x)$ is the dual value function given by
					\bqn
						v(y;x):=\inf_{\meq\in\q}G\Big(\meq,v_\meq(y)+xy\Big).\label{vyx}
					\eqn}
				\item[ii)]{Moreover, if $v(y;x)<\infty$, then the dual problem \eqref{vyx} admits a solution $(\widehat\meq,\widehat Y)$ that is maximal in the sense that any other solution $(\meq,Y)$ satisfies $\meq\ll\widehat\meq$ and $Y_T/Z=\widehat Y_T/\widehat Z$, $\meq$-a.s. }
			\end{itemize}
		\end{thm}

			While the above result imposes stronger conditions on the structural properties of $G$ (and $U$) than needed for Theorem \ref{existence}, finiteness assumptions of the type needed for the existence are not required here. This should be related to the results in \cite{kramkov,kramkov03}, where the conjugacy relations are proven under much weaker conditions than needed for the existence.\label{discinterchange}

			For the variational case, the relation between the primal and dual value function (cf. \eqref{udual}) was established in \cite{schied} under the additional assumptions that $u_{\meq_0}(x)<\infty$ for some $\meq_0\in\q_e$ and
			\bqn
				\textrm{$v(y;x)<\infty$ implies $v_{\meq_1}(y)<\infty$ for some $\meq_1\in\q_e$.}\label{additional}
			\eqn
			Although very natural, this assumption is, in fact, not needed. However, for the variational case, a slightly stronger result can be proven under this additional assumption; cf. Remark \ref{normaldual} below.


	The next result relates the respective solutions to the primal and dual problems. In particular, the result yields the existence of an equivalent auxiliary problem and results on the uniqueness of the optimal strategy.\\

		\begin{thm}\label{relation}
			Let $G\in\g$ and assume Assumption \ref{dualass} and the assumptions of Theorem \ref{existence} to hold. Let $\bar X$ be a solution to the primal problem. Then, the primal problem admits a saddle point $(\bar X,\bar \meq)$ and there exists $y^*>0$ for which the infimum in \eqref{udual} is attained.
			
			Assume, furthermore, that $G(\meq,\cdot)$, $\meq\in\q$, is strictly increasing and let $(\widehat Y,\widehat\meq)$ be any solution to the dual problem at level $y^*$. Then, $(\bar X, \widehat\meq)$ is a saddle point for the primal problem and, moreover,
				\bqn
					\bar X=I(\widehat Y_T/\widehat Z),\quad \textrm{$\widehat\meq$-a.s.}\label{relpp}
				\eqn
		\end{thm}

		The above result is of particular interest as it gives a further understanding of the optimal strategy. Let $\widehat\meq$ be the measure in the maximal solution to the dual (or, equivalently, primal) problem. Following the related literature, we refer to it as the least favourable measure. Relation \eqref{relpp} then implies that the optimal solution $\widehat X_T$, in fact, is $\widehat\meq$-a.s. unique. Consequently, if the least favourable measure is equivalent to $\mep$, the solution is $\mep$-a.s. unique\label{uniqueness}. In particular, it can then be recovered from the dual solution. In general, the least favourable measure need however not be equivalent to $\mep$. Nevertheless, an optimal strategy can still be constructed from a given solution of the dual problem by superhedging of an appropriate claim (cf. Corollary 2.7 in \cite{schied}).

		The existence of a saddle point also implies that, \emph{a posteriori}, there exists an auxiliary investment problem \eqref{primal} producing the same optimal behaviour as the original criterion (in particular, this always holds for the variational criteria). Specifically, the auxiliary problem defined with respect to the least favourable measure $\widehat\meq$ admits a solution (cf. Lemma \ref{key} below) which $\widehat\meq$-a.s. coincides with the solution to the original problem. Since the least-favourable measure is part of the solution to the original problem, the equivalence between the ambiguity averse and the auxiliary problem is an a posteriori result. Given the existence of an auxiliary problem (for which the objective function is concave), the $\widehat\meq$-a.s. uniqueness of the optimal strategy is natural. Under the weaker assumptions of Theorem \ref{existence}, it is not clear though whether an equivalent auxiliary problem exists (cf. the discussion on page\label{auxreplace} \pageref{equivaux}). The question of to which extent uniqueness holds under more general assumptions (given $G(\meq,\cdot)$ strictly increasing) is left for future study.  \\

	\begin{rem}[Time consistency]\label{tcrm}
		In general, Problem \ref{problem4} is not a time-consistent investment problem (see \cite{schied} for counter-examples in the variational case). A natural question is therefore under what assumptions it is. Indeed, while time-consistency is of interest in its own right, it is also of essential importance as it enables the use of stochastic control methods. In consequence, it is a prerequisite for extending to the quasiconcave case the explicit results obtained in terms of PDEs and BSDEs for the variational preferences (cf. the references in the Introduction). Time consistency of quasiconcave utility functionals (quasiconvex risk measures) remains, however, an open problem and, in particular, feasible explicit examples are few. Indeed, while necessary and sufficient conditions for temporal consistency of convex risk measures were established in \cite{delbaen10} (see also \cite{bionnadal06} and \cite{irina}), such results are, yet, lacking for the quasiconvex
		case. In recent work, \cite{fritelli11,fritelli12} initiated such a programme via the study of conditional quasiconvex risk measures. Questions of temporal consistency and the relation to $g$-expectations have also been studied within the more general framework of non-linear expectations in \cite{peng05}. 


		We also note that in \eqref{u}, the risk preferences of the investor are modelled via a standard continuous and concave utility function. While this assumption makes perfect sense for the static problem, the value function $u(x)$ will only satisfy weaker properties (a more precise study of the properties of the value function is left for future study). Hence, a study of time-consistent problems of this type also requires a study of the risk and ambiguity averse investment problem \eqref{u} under weaker assumptions on $U(x)$. In summary, while questions of time-consistency are of great interest, they impose challenging additional problems which are left for future study. 
	\end{rem}

	\section{Proofs}\label{proof4} 
	
		The above theorems extend results established for the variational preferences in \cite{schied} (cf. \cite{quenez04,wu} for the coherent case) to the quasiconcave risk- and ambiguity-averse investment criteria. Naturally, our proofs are therefore inspired by and in many ways similar to those in the former articles. Specifically, also here the idea is to establish results for the risk- and ambiguity-averse problem by relying on existing results for the auxiliary problem \eqref{primal}. As pointed out above, for $\meq\sim\mep$, this is the classical utility maximization problem as studied in \cite{kramkov,kramkov03}. For $\meq\ll\mep$, this need not be the case. In \cite{schied}, this is dealt with by use of certain limiting arguments. However, for our case, the weaker properties of $G$ imply that this approach does not apply. Hence the need for somewhat different arguments.

		In Section \ref{aux}, we explain why we were not able to directly apply the arguments developed in \cite{schied} and describe the approach we use in further detail. This illustrates the differences in the required assumptions. The existence and duality results are proven, respectively, in Sections \ref{existproof} and \ref{dualityproof}. The proof of Lemma \ref{key} below is deferred to the Appendix.


	\subsection{The auxiliary problem and its significance}\label{aux}

	For the variational case studied in \cite{schied}, the appearance of measures in $\q$ which are absolutely continuous but not necessarily equivalent to $\mep$, is addressed in the following way. For $\meq_0\in\q\setminus\q_e$, the authors let $\meq_1\in\q_e$ and define the measures $\meq_t$, $t\in[0,1]$, via the Radon-Nikodym derivative $Z_t:=(1-t)Z_0+tZ_1$, where $Z_0$ and $Z_1$ correspond, respectively, to $\meq_0$ and $\meq_1$. Then, $\meq_t\in\q_e$ for $t\in(0,1]$. Moreover, under suitable finiteness-assumptions, the functions
	\bqn
		t\to G\big(\meq_t,u_{\meq_t}(x)\big)\;\textrm{ and }\;  t\to G\big(\meq_t,v_{\meq_t}(y)+xy\big),\label{schiedcase}
	\eqn
	are continuous and upper semicontinuous, respectively. Combined, this implies that the results can be established by relying on results for auxiliary problems (cf. \eqref{primal}) defined with respect to measures equivalent to $\mep$ only. To the latter, the results in \cite{kramkov,kramkov03} can, in turn, be directly applied. In consequence, certain assumptions need only be posed on the set of measures $\q_e$ as opposed to $\q$ (cf. Theorem \ref{existence}). This approach is closely related to the fact that the set of absolutely continuous measures in the representation of a concave utility functional, under the assumption that there exists an equivalent measure for which the penalty is finite, can be replaced by the equivalent ones; cf. \cite{kloppel05}.

	The continuity properties in \eqref{schiedcase} rely, however, on the fact that for the variational case $G(\cdot,t)=\gamma(\cdot)+t$ and, thus, the mapping is convex. Indeed, this implies that the mapping 
		\bqn
			t\to G(Z_t,s), \quad t\in [0,1], s\in\R,\label{whynot}
		\eqn
	is convex as well. If it is finite, it is therefore upper-semicontinuous. According to Lemma 3.3 in \cite{wu}, $t\to u_{Z_t}(x)$ is continuous and, under rather weak assumptions, also $t\to v_{Z_t}(y)$ is upper semicontinuous. Hence, (under suitable finiteness assumptions) the required continuity properties in \eqref{schiedcase} follow. The crucial point is that for the quasiconcave case we consider, the function $G(\cdot,t)$ might not be convex. Therefore, the continuity properties in \eqref{schiedcase} might not hold and, in consequence, the arguments developed in \cite{schied} do not apply.

	Our methods are rather based on a closer study of the auxiliary problem \eqref{primal}. Although not a standard problem, it is, in fact, only a minor modification thereof and can therefore be solved by use of the same methods as developed in \cite{kramkov} (cf. Lemma \ref{key} below). By use of this observation, we then address the risk and ambiguity averse investment problem. In particular, we obtain alternative proofs of some results established in \cite{schied}. Indeed, once the relevant properties have been obtained for auxiliary problems with $\meq\ll\mep$, the proofs can be simplified. In particular, our approach enables establishing the relation between the primal and dual value functions (cf. \eqref{udual} below) under slightly weaker assumptions than in \cite{schied}. However, for the existence of an optimizer, stronger assumptions are required. These assumptions, which are a consequence of the more general type of criteria we consider, are still economically feasible; see the discussion after Theorem \ref{existence} above.

	The properties of the auxiliary problem that will be made use of are presented next.\\
			
		\begin{lem}\label{key}
			Let $U$ satisfy the Inada conditions (cf. \eqref{inadaq}) and assume $u_Z(x_0)<\infty$ for some $x_0>0$. Then, the function $u_Z(x)$ and $v_Z(y)$ defined, respectively, in \eqref{primal} and \eqref{dual} satisfy the following: 
			\begin{itemize}
				\item[i)]{It holds that
				\bq
					v_Z(y)=\sup_{x>0}\big(u_Z(x)-xy\big)
					\;\;\textrm{and}\;\;
					u_Z(y)=\inf_{y>0}\big(v_Z(x)+xy\big),
				\eq
			and, furthermore, 
				\bq
					u_Z'(0)=\infty\;\;\textrm{and}\;\; v_Z'(\infty)=0.
				\eq}
				\item[ii)]{Under the additional assumption that $v_Z(y)<\infty$, $y>0$, it also holds that
				\bq
					u_Z'(\infty)=0\;\;\textrm{and}\;\; v_Z'(0)=-\infty.
				\eq
			Moreover, the set $\{ZU^+(g):g\in\cmc(x)\}$ is $\mep$-UI and the primal problem admits a solution.}
			\end{itemize} 
		\end{lem}

	For $\meq\sim\mep$, the above result was established in \cite{kramkov,kramkov03}. For $\meq\ll\mep$, note that while the objective in the auxiliary problem is defined with respect to the measure $\meq$, the set of strategies is still defined in the standard way with respect to the reference measure $\mep$. This is the key reason for the above result to hold also for auxiliary problems defined with respect to $\meq\ll\mep$ (cf. also Remark \ref{normaldual}). Indeed, the proofs in \cite{kramkov,kramkov03} make use of the assumption of NFLVR in order to treat the strategies (via the characterization in \eqref{iff} below), not the objective function. In consequence, Lemma \ref{key} follows by minor modifications the respective proofs in \cite{kramkov,kramkov03}. For completeness, the details are presented in the Appendix. We note that the auxiliary problem also might be viewed as utility maximization under $\mep$ with respect to the stochastic utility function $\tilde U(x)=Z_TU(x)$ and refer to, among others, \cite{follmer2000efficient} for alternative arguments. For an economic interpretation of the condition that $u_Z(x_0)<\infty$, see the discussion on p. \pageref{kardaras}.

	\subsection{Proof of the existence of an optimal investment strategy}\label{existproof}
	
		The proof first establishes upper semicontinuity and quasi-concavity of the objective function. The existence of an optimizer then follows by use of a Komlos-type argument.

		\begin{proof}[Proof of Theorem \ref{existence}]
			For the case when $G(\cdot,t)$ is convex, we refer to Lemma 4.7 in \cite{schied}. Note that the optimization over $X\in\xmc(x)$ can be replaced by optimization over the set $\cmc(x):=\{g\in L^0_+(\f_T):g\le X_T, X\in\mathcal{X}(x)\}$ of random variables. That is to say,
			\bqn
				u(x)=\sup_{g\in\cmc(x)}\inf_{\meq\in\q}G\Big(\meq, \e^\meq\big[U\big(g\big)\big]\Big).\label{ug}
			\eqn

			Due to assumption \eqref{aspt}, $u_\meq(x)<\infty$, $\meq\in\meq$, (cf. Lemma 3.5 in \cite{wu}) and Lemma \ref{key} applies. Hence, $\{U^+(f)\}$, $f\in\cmc(x)$ is $\meq$-UI, $\meq\in\q$. Combined with an application of Fatou's Lemma to the negative part $U^-(f)$, this yields that $g\to\e^\meq[U(g)]$, $g\in\cmc(x)$, is upper semicontinuous for $\meq\in\q$. Since also $G(Z,\cdot)$ is upper semicontinuous and the pointwise infimum of u.s.c. functionals is again u.s.c., it follows that the mapping  
			\bqn
				g\to V(g):=\inf_{\meq\in\q}G\Big(\meq,\e^\meq[U(g)]\Big), \quad g\in \cmc(x),
			\eqn
			is upper semicontinuous as well. Moreover, as $g\to\e[ZU(f)]$ is concave, $G(Z,\cdot)$ is quasi-linear and the point-wise infimum of quasiconcave functionals is again quasiconcave, it follows that $g\to V(g)$ is quasiconcave. More precisely, the semicontinuity and quasi-concavity can be argued as follows. Since $t\to G(q,t)$ is non-decreasing and right-continuous, it holds according to Proposition B.2 in \cite{drapeau10} that
				\bq
					G(q,t)\ge m\quad\textrm{if and only if}\quad t\ge G^{(-1,l)}(q,m),
				\eq		
			where the left inverse is given by $G^{(-1,l)}(q,m)=\inf\{n\in\R:G(q,m)\ge n\}$. Consequently, the following holds:
				\bqn
					\Big\{g\in\cmc(x):G\big(q,\langle U(g),q\rangle\big)\ge m\Big\}
					=\Big\{g\in\cmc(x):\langle q,U(g)\rangle\ge G^{(-1,l)}(q,m)\Big\}.\label{cf}
				\eqn
			This level set is then closed and convex since $g\to\langle q,U(g)\rangle$, $g\in\cmc(x)$, $q\in\q$, is concave and upper semicontinuous according to the above. Then, in turn, it follows that the set
				\bq
					\Big\{g\in\cmc(x):V(g)\ge m\Big\}
					=\cap_{q\in\q}\Big\{g\in\cmc(x):G\big(q,\langle U(g),q\rangle\big)\ge m\Big\},
				\eq
			is closed and convex and, thus, that $g\to V(g)$ is upper semicontinuous and quasiconcave.

			The existence of an optimizer now follows by a standard Komlos-type argument. Indeed, let $(g_n)$ be a sequence in $\cmc(x)$ such that $V(g_n)\nearrow u(x)$. Since each $g_n\ge 0$, Komlos Lemma gives $\tilde g_n\in\textrm{conv}(g_n,g_{n+1},...)$ converging $\mep$-a.s. to some $g$. Since $\cmc(x)$ is convex, $\tilde g_n\in \cmc(x)$. Furthermore, according to Proposition 3.1 in \cite{kramkov}, $\cmc(x)$ is closed under convergence in $L^0$ and, thus, $g\in\cmc(x)$. Next, due to the quasi-concavity of $g\to V(g)$ it holds that
				\bq
					V\left(\tilde g_n\right)=V\bigg(\sum_{k\ge n}\lambda^kg_k\bigg)\ge \inf_{k\ge n} V(g_k)=V(g_n). 
				\eq
			Since $V(\tilde g_n)\le u(x)$ and $V(g_n)\nearrow u(x)$ it, thus, follows that $\lim_{n\to\infty}V(\tilde g_n)=u(x)$; that is to say, $\tilde g_n$ is an optimizing sequence. Then, in turn, the upper semi-continuity of $g\to V(g)$ yields that
				\bq
					V(g)\ge \limsup_{n\to\infty} V(\tilde g_n)=u(x), 
				\eq
			which concludes the proof. 			
		\end{proof}

	\subsection{Proof of the duality results}\label{dualityproof}

		For the proof of Theorem \ref{interchange}, note that since $G(Z,\cdot)$ is non-decreasing, it immediately follows that
			\bqn
				u(x)\le \inf_{\meq\in\q}G\big(\meq,u_{\meq}(x)\big)
				\le \inf_{\meq\in\q}G\big(\meq,v_{\meq}(y)+xy\big)=v(y;x), \quad \textrm{ for all}\; y>0. \label{proo}
			\eqn
		The proof verifies that the inequalities, in fact, are equalities by use of, respectively, Sion's minimax theorem and Lemma \ref{key}.

			\begin{proof}[Proof of Theorem \ref{interchange}]
			\emph{Part i)} For the case when $G(Z,\cdot)$ is concave and $U:\R_+\to\R_+$ let, respectively, $\varepsilon>0$ and $\varepsilon=0$. According to \eqref{finiteness}, $u(x)<\infty$, $x>0$. Hence, for each $g\in\cmc(x)$, 
			\bq
				\inf_{Z\in\q}G\big(Z,\e\big[ZU\left(\varepsilon+g\right)\big]\big)
				\le 
				\sup_{g\in\cmc(x)}\inf_{Z\in\q}G\big(Z,\e\big[ZU\left(\varepsilon+g\right)\big]\big)
				\le u(x+\varepsilon)<\infty.
			\eq
		Consequently,
			\bqn
				\inf_{Z\in\q}G\big(Z,\e\big[ZU\left(\varepsilon+g\right)\big]\big)
				=\inf_{Z\in\tilde \q}G\big(Z,\e\big[ZU\left(\varepsilon+g\right)\big]\big),\label{framtid}
			\eqn
		where $\tilde\q$ is the set of measures in $\q$ for which $G\big(Z,\e[ZU(\varepsilon+g)]\big)\le u(x+\varepsilon)+1$. Since $\e[ZU(\varepsilon+g)]\ge U(\varepsilon)\wedge 0$, for $Z\in\q$, it holds for each $Z\in\tilde\q$ that $G(Z,t)\le c$ for $t:=U(\varepsilon)\wedge 0$ and $c:=u(x+\varepsilon)+1$. Hence, $\tilde\q$ in \eqref{framtid} might be replaced by $\q_t(c)$. Since this holds for each $g\in\cmc(x)$, it thus follows that
			\bqn
				\sup_{g\in\cmc(x)}\inf_{Z\in\q}G\big(Z,\e\big[ZU\left(\varepsilon+g\right)\big]\big)
				=\sup_{g\in\cmc(x)}\inf_{Z\in\q_t(c)}G\big(Z,\e\big[ZU\left(\varepsilon+g\right)\big]\big).\label{gard}
			\eqn

			Next, as $U(\varepsilon+\cdot)$ is bounded from below and $\q_t(c)$ is UI due to Assumption \ref{dualass}, an application of Fatou's Lemma yields $Z\to\e[ZU(\varepsilon+g)]$, $g\in\cmc(x)$, lower semicontinuous with respect to a.s. convergence on $\q_t(c)$. As $\q_{t,T}$ is UI, that is equivalent to lower semicontinuity with respect to convergence in $L^1$. As the functional is convex (affine) that, in turn, implies weak lower semicontinuity. Since $G$ is jointly lower semicontinuous and quasiconvex it, thus, follows that
				\bqn
					Z\to G\big(Z,\e\big[ZU\left(\varepsilon+g\right)\big]\big),\quad g\in\cmc(x),     
  					     \label{semi}
				\eqn
			is weakly lower semicontinuous and quasiconvex. Furthermore, as established above, it holds that 
				\bqn
					g\to G\big(Z,\e\big[ZU\left(\varepsilon+g\right)\big]\big),\quad Z\in\q,\label{quasiconcave}
				\eqn
			is quasiconcave. Recall that $\cmc(x)$ is convex. Moreover, 
				\bq
					G(\lambda Z+(1-\lambda)\bar Z,t)\le \max\{G(Z,t),G(\bar Z,t)\}\le c,
				\eq 
			for $Z,\bar Z\in\q_t(c)$. Hence, $\q_t(c)$ is also convex. The latter set is also weakly compact due to Assumption \ref{dualass}. Given the properties of the mappings defined in \eqref{semi} and \eqref{quasiconcave}, respectively, we might thus apply Sion's minimax theorem (cf. \cite{sion58}). This yields, 
				\bqn
					\sup_{g\in\cmc(x)}\inf_{Z\in\q_t(c)}G\big(Z,\e\big[ZU\left(\varepsilon+g\right)\big]\big)
					=\inf_{Z\in\q_t(c)}\sup_{g\in\cmc(x)}G\big(Z,\e\big[ZU\left(\varepsilon+g\right)\big]\big).\label{sion}
				\eqn
		Note that 
			\bq
				\inf_{Z\in\q}\sup_{g\in\cmc(x)}G\big(Z,\e\big[ZU\left(\varepsilon+g\right)\big]\big)
				\le 
				\inf_{Z\in\q_t(c)}\sup_{g\in\cmc(x)}G\big(Z,\e\big[ZU\left(\varepsilon+g\right)\big]\big)\le u(x+\varepsilon),
			\eq
		where the first inequality is trivial and the second follows from \eqref{sion} combined with \eqref{gard}. Hence, by use of the same argument as above, the set $\q_t(c)$ in the right hand side of \eqref{sion} can be substituted for the set $\q$. That is to say, 
				\bqn
					\sup_{g\in\cmc(x)}\inf_{Z\in\q}G\big(Z,\e\big[ZU\left(\varepsilon+g\right)\big]\big)
					=\inf_{Z\in\q}\sup_{g\in\cmc(x)}G\big(Z,\e\big[ZU\left(\varepsilon+g\right)\big]\big).\label{sion2}
				\eqn
			For the case when $U:\R_+\to\R_+$, this completes the proof of part i).

			For the case when $G(Z,\cdot)$ is concave, a straight-forward argument yields $u(x)$ concave. According to \eqref{finiteness} it is also finite. Consequently, it is continuous as a concave function is continuous on the interior of the set where it is finite (cf. Theorem 10.1 in \cite{rock70}). On the other hand, since the expression on the left-hand-side in \eqref{sion2} clearly is smaller than $u(x+\varepsilon)$, it follows that
				\bqqn
					u(x+\varepsilon)
					&\ge &\inf_{Z\in\q}\sup_{g\in\cmc(x)}G\big(Z,\e\big[ZU\left(g\right)\big]\big)\nn\\
					&\ge &\sup_{g\in\cmc(x)}\inf_{Z\in\q}G\big(Z,\e\big[ZU\left(g\right)\big]\big)\label{magic}
					\;=\;u(x).
				\eqqn
			Since $u(x)$ is (upper semi-) continuous, the result then follows by letting $\varepsilon\searrow 0$.

			Next, according to Assumption \ref{standing}, $\q^f\neq\emptyset$. Since $G$ has an asymptotic maximum in the sense of Definition \ref{defg}, this implies that the set $\q$ on the right-hand side in \eqref{hoppas} can be replaced by the set $\q^f$. For $\meq\in\q^f$, Lemma \ref{key} applies. Hence, use of part i), the fact that $G(Z,\cdot)$ is non-decreasing and the duality relations between $u_Z(x)$ and $v_Z(y)$ given in Lemma \ref{key}, yields
				\bqq
					u(x)&=&
					\inf_{Z\in\q^f}G\Big(Z,u_Z(x)\Big)\\
					&=&\inf_{Z\in\q^f}G\Big(Z,\inf_{y>0}\left(v_Z(y)+xy\right)\Big)
					\;=\;\inf_{y>0}\inf_{Z\in\q^f}G\Big(Z,v_Z(y)+xy\Big).
				\eqq
			Given the definition of $v(y;x)$, it thus only remains to show that 
				\bqn
					\inf_{Z\in\q^f}G\Big(Z,v_Z(y)+xy\Big)=\inf_{Z\in\q}G\Big(Z,v_Z(y)+xy\Big)=:v(y;x).\label{prokl}
				\eqn
			The inequality "$\ge$" follows as $\q^f\subseteq\q$. Without loss of generality, assume $v(y;x)<\infty$ and let $\tilde\q:=\{Z\in\q:v_Z(y)<\infty)\}$. Clearly the set $\q$ on the right-hand side of \eqref{prokl} can then be replaced by $\tilde \q$. On the other hand, $v_Z(y)<\infty$ implies that $u_Z(x)<\infty$. Hence, $\tilde \q\subseteq \q^f$ which yields the inequality "$\le$".

			\emph{Part ii)} Let
			\bq
				H(Z,h):=G(Z,\e[ZV(h/Z)]+xy).
			\eq
			According to Lemma 3.7 in \cite{wu}, $(Z,h)\to\e[ZV(h/Z)]$ is lower semicontinuous. Hence, since $G$ is jointly lower semicontinuous, it follows that so is $(Z,h)\to H(Z,h)$. Furthermore, since $(z,y)\to zV(y/z)$ is convex, $G(Z,\cdot)$ non-decreasing and $G$ jointly quasiconvex, it follows that $(Z,h)\to H(Z,h)$ is jointly quasiconvex. Indeed, let $Z_t=tZ_0+(1-t)Z_1$ and $h_t=th_0+(1-t)h_1$. Then, 
				\bqq
				 H(Z_t,h_t)&=&
				 G(Z_t,\e[Z_tV(h_t/Z_t)]+xy)\\
				 &\le & G(Z_t,t\e[Z_0V(h_0/Z_0)]+(1-t)\e[Z_1V(h_1/Z_1)]+xy)
				 \;\le \; H(Z_0,h_0)\vee H(Z_1,h_1). 
				\eqq		
			Let $(Z_n,h_n)\in\q\times\dmc(y)$ be an optimizing sequence such that
			\bqn
				G\Big(Z^n,\e\left[Z^nV\left(h^n/Z^n\right)\right]+xy\Big)\sumpport_{n\to\infty}v(y;x)<\infty.\label{mot1}
			\eqn
		Note that 
			\bqn
				\e[ZV(h/Z)]+xy \ge U(x), \label{mot2}
			\eqn
		for all $Z\in\q$ and $h\in\mathcal{D}(y)$. Indeed, as $V$ is convex, use of Jensen's inequality yields
			\bq
				\e\big[ZV\left(h/Z\right)\big]
				=\e\big[ZV\left(h/Z\ind_{Z>0}\right)\big] 
				\ge V\big(\e\left[Zh/Z\ind_{Z>0}\right]\big)
				=V\big(\e\left[h\ind_{Z>0}\right]\big)
				\ge V(y),
			\eq
		as $V$ is decreasing and $h\in\dmc(y)$ which implies that $\e[h\ind_{Z>0}]\le \e[h]\le y$. The fact that $G(Z,\cdot)$ is non-decreasing combined with \eqref{mot1} and \eqref{mot2}, then yields
			\bq
				\tilde c:=\limsup_{n\to\infty}G\Big(Z^n,U(x)\Big)<\infty
			\eq
		and w.l.o.g., we can assume that $Z^n\in\{Z\in\q:G(Z,U(x))\le \tilde c+1\}$. That is to say, that $Z^n\in\q_t(c)$, for $t:=U(x)$ and $c:=\tilde c+1$.

		Applying twice the Komlos Lemma, yields a sequence $(\tilde Z_n,\tilde h_n)\in\textrm{conv}\{(Z_n,h_n), (Z_{n+1},h_{n+1}),...\}$ which converges $\mep$-a.s. to some $(Z_0,h_0)$. Since $\q_t(c)$ and $\mathcal{D}(y)$ both are convex, $(\tilde Z_n,\tilde h_n)\in\q_t(c)\times\mathcal{D}(y)$. Furthermore, since $\q_t(c)$ is uniformly integrable (as it is weakly compact) and $\mathcal{D}(y)$ is closed in $L^0$ according to Proposition 3.1 in \cite{kramkov}, it follows that $(Z_0,h_0)\in\q_t(c)\times\mathcal{D}(y)$. Moreover, use of the quasiconvexity yields
				\bq
					H(\tilde Z_n,\tilde h_n) = H\bigg(\sum_{k\ge n}\lambda^kZ_k,\sum_{k\ge n}\lambda^kh_k\bigg)
					\le \max_{k\ge n} H(Z_k,h_k)=H(Z_n,h_n)\searrow v(y). 
				\eq
			Hence, $(\tilde Z_n,\tilde h_n)$ is also an optimizing sequence. Consequently, by use of the lower semicontinuity, it follows that
				\bq
					H(Z_0,h_0)\le \liminf_{n\to\infty} H(\tilde Z_n,\tilde h_n)=v(y),
				\eq
			which proves that the optimum is attained for $(Z_0,h_0)$.

			Next, suppose $(\tilde Z_1,\tilde h_1)$ is another optimal pair. Let $h_t:=t\tilde h_1+(1-t)\tilde h_0$ and $Z_t:=t\tilde Z_1+(1-t)\tilde Z_0$, $t\in[0,1]$. As $(Z,h)\to\e[ZV(h/Z)]$ is convex and $G$ jointly quasiconvex, it thus follows that				
		{\setlength{\arraycolsep}{-0.5cm}
				\bqq
					G(Z_t,\e\left[Z_tV\left(h_t/Z_t\right)\right]+xy)\\
					&&\le G\left(Z_t,t\e\left[Z_1V\left(h_1/Z_1\right)\right]+(1-t)\e\left[Z_0V\left(h_0/Z_0\right)\right]+xy\right)\\
					&&\le  \max\big \{G\left(Z_1,\e\left[Z_1V\left(h_1/Z_1\right)\right]+xy\right),G\left(Z_0,\e\left[Z_0V\left(h_0/Z_0\right)\right]+xy\right)\big\}
					\;=\; v(y),
				\eqq}
			due to the optimality of $(\tilde Z_1,\tilde h_1)$ and $(\tilde Z_0,\tilde h_0)$, respectively. Hence, also $(h_t,Z_t)$ is optimal. We proceed as in the proof of Lemma 4.3 in \cite{schied}. Note that for $t\in(0,1)$, $\{Z_t>0\}=\{Z_0>0\}\cup\{Z_1>0\}$. Also note that according to (25) in \cite{schied}, the ration $h_t/Z_t$, does not depend on $t$. Hence, there exists a random variable $Y_T\ge 0$ and a sequence $\bar Z_1$, $\bar Z_2$, ... such that:
				\begin{itemize}
					\item[(a)]{$\mep[\bar Z_n]$ tends to the maximum $\mep$-probability for the support of any optimal $Z$;}
					\item[(b)]{$\{\bar Z_1>0\}\subseteq \{\bar Z_2>0\}\subseteq ...$;}
					\item[(c)]{for each $n$, $\bar h_n:=Y_T\bar Z_n\in\dmc(y)$ and $(\bar h_n,\bar Z_n)$ is optimal.}
				\end{itemize}
			By use of a Komlos-type argument, we may assume that $\bar Z_n$ converge $\mep$-a.s. to some $\bar Z\in\q$. Then $\bar h:=Y_T\bar Z\in\dmc(y)$ by Proposition 3.1 in \cite{kramkov} (cf. \eqref{iff} below). As above, it follows that $(\bar h,\bar Z)$ is optimal and, clearly, it is maximal. 									
		\end{proof}

	Next, we prove Theorem \ref{relation} which establishes the existence of a saddle point and the link between the primal and dual solutions.

	\begin{proof}[Proof of Theorem \ref{relation}]		
			\emph{Part i)} In order to verify the first statement, it remains to show that the infimum on the right hand side in \eqref{hoppas} is attained. To this end, note that since $U\ge 0$, $G$ is jointly lower semicontinuous and quasiconvex and the operation of point-wise supremum preserves lower semicontinuity and quasiconvexity, it follows that
			\bqn
				Z\to \sup_{g\in\cmc(x)}G\Big(Z,\e\big[ZU(g)\big]\Big),\label{sist}
			\eqn
			is lower semicontinuous and quasiconvex. Next, let $Z^n\in\q$ be a sequence such that			
				\bqn
					G\Big(Z^n,u_{Z^n}(x)\Big)\searrow u(x). 
				\eqn
			As $u_{Z^n}(x)\ge U(x)$, for all $n$ and $G(Z,\cdot)$ is increasing, it follows that
				\bq
					\tilde c:=\limsup_{n\to\infty}G\big(Z^n,U(x)\big)<\infty. 
				\eq
			Hence, w.l.o.g., we may assume $Z^n\in\q_t(c)$ with $t:=U(x)$ and $c=\tilde c+1$. Application of Komlos Lemma, yields a sequence $\tilde Z_n\in\textrm{conv}\{Z_n,Z_{n+1},...\}$ which converges $\mep$-a.s. to some $Z_0$. Since $\q_t(c)$ is convex, $\tilde Z_n\in\q_t(c)$, $n=1,2,...$. Furthermore, since $\q_t(c)$ is uniformly integrable (as it is weakly compact), it follows that $Z_0\in\q_t(c)$. Moreover, as in the proof of Theorem \ref{interchange} part iii), the quasiconvexity implies that also $Z^n$ is an optimizing sequence. Use of the lower-semicontinuity of the mapping in \eqref{sist}, then yields that the infimum is attained for $Z_0$. 
		
			\emph{Part ii)} Let $\tilde g$ and $\widetilde Z$ be a saddle-point for the primal problem. Let $y^*>0$ such that the infimum for the auxiliary conjugacy relations with respect to $\widetilde Z$ are attained for $y^*$. Then, it follows that
				\bq
					u(x)=G\Big(\widetilde Z,u_{\widetilde Z}(x)\Big)=G\Big(\widetilde Z,v_{\widetilde Z}(y^*)+xy^*\Big).
				\eq
			Hence, the infimum in \eqref{udual} is attained for $y^*$. Next, let $(\widehat Z,\hat g)$ be any solution to the dual problem corresponding to $y^*$. Since $G(Z,\cdot)$ is non-decreasing, use of the assumptions and Lemma \ref{key} yields
				\bqn
					u(x)=v(y^*;x)=G\left(\hat Z,v_{\hat Z}(y^*)+xy^*\right)\ge G\left(\hat Z,u_{\hat Z}(x)\right)\ge u(x).\label{sol}
				\eqn
			Hence, we have equality which, in turn, implies that 
				\bq
					u(x)=G\left(\hat Z,u_{\hat Z}(x)\right)
					\ge G\left(\hat Z,\e\Big[\hat ZU(\hat X)\Big]\right)
					\ge \inf_{Z\in\q}G\left(Z,\e\Big[ZU(\hat X)\Big]\right)= u(x).
				\eq
			Consequently, $(\widehat\meq,\widehat X)$ is a saddle point for the primal problem. 
			
			Next, from the definition of $\ymc(y)$, it follows that $\e[X_TY_T]\le xy$, for $X\in\xmc(x)$ and $Y\in\ymc(y)$. Thus, it follows that
				\bqq
					G\left(\hat Z,\e\Big[\hat ZV(\hat Y/\hat Z)+\hat X\hat Y\Big]\right)-G\left(\hat Z,\e\left[\hat ZU(\hat X)\right]\right)
					&=&G\left(\hat Z,\e\left[\hat ZV(\hat Y/\hat Z)\right]+\e\left[\hat X\hat Y\right]\right)-u(x)\\
					&\le &G\left(\hat Z,\e\left[\hat ZV(\hat Y/\hat Z)\right]+xy\right)-u(x)\\
					&=&v(y^*;x)-u(x)\;=\;0.
				\eqq

			Since $G(Z,\cdot)$ strictly increasing, this implies that 
				\bq
					\e^{\hat\meq}\Big[V(\hat Y/\hat Z)+\hat X\hat Y/\hat Z\Big]\le \e^{\hat\meq}\Big[U(\hat X)\Big].
				\eq
			On the other hand $V(\hat Y/\hat Z)+\hat X\hat Y/\hat Z\ge U(\hat X)$, $\hat\meq$-a.s. Consequently, $V(\hat Y/\hat Z)+\hat X\hat Y/\hat Z= U(\hat X)$ $\hat\meq$-a.s. and, thus, it follows that $\hat X=I(\hat Y/\hat Z)$, $\hat\meq$-a.s.
		\end{proof}

	\subsection{Further remarks on the auxiliary problem}

	We conclude with some further remarks on the auxiliary value functions.\\
	
	\begin{rem}\label{normaldual1}
		Let $\mathcal{\hat Y}(y)$ the set of all positive $\meq$-supermartingales such that $Y_0=y$ and $XY$ is a $\meq$-supermartingale for all $X\in\mathcal{X}(1)$. Then, as shown in Lemma 4.2 in \cite{schied}, it holds that
				\bq
					v_\meq(y)=\inf_{Y\in\mathcal{\hat Y}(y)}\e^\meq\big[V(Y_T)\big].
				\eq	
		Indeed, let $0\le s\le t\le T$. For $\hat Y\in\mathcal{\hat Y}(y)$ and $X\in\xmc(1)$,
			\bq
				X_s\hat Y_s\ge \e^\meq\left[X_t\hat Y_t|\f_s\right]=\frac{1}{Z_s}\e\left[X_t\hat Y_tZ_t|\f_s\right],\quad \textrm{$\mep$-a.s. on $\{Z_s>0\}$}. 
			\eq
		On $\{Z_s=0\}$, it holds that $Z_t=0$ $\mep$-a.s. It follows that $X\hat Y Z$ is a $\mep$-supermartingale and hence that $\hat YZ\in\ymc(y)$. Conversely, let $Y\in\ymc(y)$. Then, $\meq$-a.s. for each $X\in\xmc(1)$, 
			\bq
				\e^\meq\left[X_t\frac{Y_t}{Z_t}|\f_s\right]
				=\frac{1}{Z_s}\e\left[Z_tX_t\frac{Y_t}{Z_t}\ind_{Z_t>0}|\f_s\right]
				\le\frac{1}{Z_s}\e\big[X_tY_t\ind_{Z_s>0}|\f_s\big]\le \frac{X_sY_s}{Z_s}\ind_{Z_s>0}.
			\eq 
		Hence, $XY/Z$ is a $\meq$-supermartingale, for all $x\in\xmc(x)$ and $Y/Z\ind_{Z>0}\in\mathcal{\hat Y}(y)$. \\
		\end{rem}

		\begin{rem}\label{normaldual}
		For $\meq\ll\mep$, let $\mathcal{X}_\meq(x)$ the set of wealth-processes such that $X_0\le x$ and $X_t\ge 0$, $\meq$-a.s. on $t\in[0,T]$ and $\mathcal{Y}_\meq(y)$ the set of all positive $\meq$-supermartingales such that $Y_0=y$ and $XY$ is a $\meq$-supermartingale for all $X\in\mathcal{X}_\meq(1)$. Then, consider the following two problems
			\bq
				\tilde u_\meq(x)=\sup_{X\in \xmc_\meq(x)}\e^\meq\big[U(X_T^\pi)\big]\quad\textrm{and}\quad
				\tilde v_\meq(y)=\inf_{Y\in\ymc_\meq(y)}\e^\meq\big[V(Y_T)\big],
			\eq
		where $\inf\emptyset:=\infty$. Note that $\tilde u(x)$ is the standard investment problem with respect to $\meq$ as it is normally defined. However, since it is not clear whether $S_t$ satisfies NFLVR with respect to $\meq$, it is (a priori) not clear whether $\tilde u$ and $\tilde v$ are each others conjugate. Also note that $\mathcal{X}(x)\subseteq \mathcal{X}_\meq(x)$, which, in turn, implies that $\ymc_\meq(y)\subseteq\mathcal{\hat Y}(y)$. Hence,
				\bqn
					u(x)\le \tilde u(x)\quad\textrm{and}\quad v_\meq(y)\le \tilde v_\meq(y).\label{remaining}
				\eqn
		In particular, the condition that $\tilde v_\meq(y)<\infty$, $\meq\in\q$, is therefore a sufficient condition for Theorem \ref{existence} to hold. For $\meq\sim\mep$, \eqref{remaining} holds with equality. The question whether there are models $\meq\ll\mep$, for which the inequality is strict is left for future study. We limit ourselves to note that given that the market is continuous, equality may in fact hold under rather weak conditions. Regardless of this, under the additional assumption \eqref{additional}, it holds for the variational case (i.e. when $G(\meq,t)=\gamma(\meq)+t$) that
				\bq
					v(y;x)=\inf_{\meq\in\q}G\Big(\meq,\tilde v_\meq(y)+xy\Big),
				\eq
		and, in consequence, Theorem \ref{interchange} holds also with $v_\meq(y)$ replaced by $\tilde v_\meq(y)$. Indeed, $v(y;x)=v(y)+xy$, where
			 	\bq
					v(y)=\inf_{Z\in\q^e}\Big(v_Z(y)+\gamma(Z)\Big)
					=\inf_{Z\in\q^e}\Big(\tilde v_Z(y)+\gamma(Z)\Big)
					\ge \inf_{Z\in\q}\Big(\tilde v_Z(y)+\gamma(Z)\Big).
				\eq
		Here the first equality follows from Lemma 4.4 in \cite{schied}. Since $v_\meq(y)\le \tilde v_\meq(y)$, $\meq\in\q$, equality follows. 	
	\end{rem}

\appendix

	\section*{Appendix}\label{keyproof}

	We here provide the proof of Lemma \ref{key}. We stress that the Lemma follows by minor modifications in the respective proofs in \cite{kramkov,kramkov03}. Specifically, in the proofs of Lemmata 3.4 and 3.5 in \cite{kramkov} and Lemma 1 in \cite{kramkov03}. For completeness, the details are presented next. For alternative arguments, see the further discussion in Section \ref{aux}. 

	In preparation for the proof, note that the assumption $u_\meq(x)<\infty$ implies that the expectation operator is defined in the standard way (cf. page \pageref{expoperator}). It also immediately follows that 
		\bq
			u_Z(x)=\sup_{g\in\cmc(x)}\e[ZU(g)]\qquad\textrm{and}\quad
			v_Z(y)=\inf_{h\in\mathcal{D}(y)}\e[ZV(h/Z)],
		\eq
	where the set of random variables $\cmc(x)$ and $\dmc(y)$ are defined by $\cmc(x):=\{g\in L^0_+:g\le X_T,\;\mep\textrm{-a.s.}, X\in\mathcal{X}(x)\}$ and $\mathcal{D}(y)=\{h\in L^0_+:h\le Z,\;\mep\textrm{-a.s.}, Z\in\mathcal{Y}(y)\}$. Moreover, according to Proposition 3.1 in \cite{kramkov}, it holds that
			\bqn
				g\in\cmc(x)\;\textrm{ if and only if }\;\e[gh]\le xy, \;\textrm{for all}\; h\in\mathcal{D}(y).\label{iff}
			\eqn	

	\begin{proof}[Proof of Lemma \ref{key}]
			Consider the mapping $\bmc_n\times\dmc (y)\to\R$, $(g,h)\to\e[ZU(g)-gh]$, where $\bmc_n:=\{g\in L^0_+:0\le g\le n\}$. The set $\dmc(y)$ is convex and since the unit ball in $L^\infty$ is weak*-compact, so is $\bmc_n$. Furthermore, while the above mapping is concave in $g$, it is linear and continuous in $h$, this for the weak*-topology on $L^0$. Hence, the minimax theorem (cf. Theorem 2.7.1 in \cite{aubin00}) can be applied in order to obtain
				\bqn
					\sup_{g\in\bmc_n}\inf_{h\in\dmc(y)}\e[ZU(g)-gh]
					=\inf_{h\in\dmc(y)}\sup_{g\in\bmc_n}\e[ZU(g)-gh].\label{1}
				\eqn
			Next, \eqref{iff} implies that
				\bqn
					\lim_{n\to\infty}\sup_{g\in\bmc_n}\inf_{h\in\dmc(y)}\e[ZU(g)-gh]
					=\sup_{x>0}\big(u_Z(x)-xy\big). \label{2}
				\eqn
			Indeed, we first see that \eqref{iff} implies that $\inf_{h\in\dmc(y)}\e[ZU(g)-gh]\ge \e[ZU(g)]-xy$. Taking the supremum over $g\in\bmc_n$ and $g\in\cmc(x)\cap\bmc_n$ on the left- and right-hand side, respectively and, then, in turn, letting $n\to\infty$ yields the inequality $\ge$ in \eqref{2}. Next, we fix $n$ and let $g\in\bmc_n$ and $x^*:=\inf\{x>0:g\in \cmc(x)\}$. Without loss of generality, let $x^*>0$. Then it holds that $g\in\cmc(x^*+\varepsilon)$ but $g\not\in\cmc(x^*-\varepsilon)$. Thus, using \eqref{iff}, it follows that
				\bqq
					\inf_{h\in\dmc(y)}\e[ZU(g)-gh]
					&<&\e[ZU(g)]-(x^*-\varepsilon)y \\
					&\le& u_Z(x^*+\varepsilon)-(x^*+\varepsilon)y+2\varepsilon y 
					\;\le\; 2\varepsilon y+\sup_{x>0}\big(u_Z(x)-xy\big). 
				\eqq
			Letting $\varepsilon\searrow 0$, yields that for $g\in\bmc_n$ and $n\in\mathbb{N}$,
			\bq
				\inf_{h\in\dmc(y)}\e[ZU(g)-gh]\le \sup_{x>0}\big(u_Z(x)-xy\big).
			\eq 
			This completes the proof of \eqref{2}. 
			
			Next, let $V_n(y)=\sup_{0\le x\le n}\big(U(x)-xy\big)$ and note that
				\bqqn
					\sup_{g\in\bmc_n}\e[ZU(g)-gh]
					&=& \e\big[\sup_{0<x\le n}\left\{ZU(x)-xh\right\}\big]\nn\\
					&=& \e\big[\sup_{0<x\le n}\left\{\left(ZU(x)-xh\right)\ind_{Z>0}\right\}\big]
					\;=\;\e[ZJ_n(h/Z)],
				\eqqn
			where it was used that $\sup_{0<x<n}\{ZU(x)-xh\}=0$ on $\{Z=0\}$. Therefore, it holds that
				\bqn
					\inf_{h\in\dmc(y)}\sup_{g\in\bmc_n}\e[ZU(g)-gh]=
					\inf_{h\in\dmc(y)}\e[ZJ_n(h/Z)]=:v_Z^n(y).\label{int}
				\eqn
			Combining \eqref{1}, \eqref{2} and \eqref{int}, we easily see that in order to show the first conjugacy relation in i), it only remains to show that
				\bqn
					\lim_{n\to\infty}v_Z^n(y)=v(y),\qquad y>0. \label{3}
				\eqn
			To this end, let $h_n\in\dmc(y)$ be a sequence such that 
				\bq
					\lim_{n\to\infty}v_Z^n(y)=\lim_{n\to\infty}\e\left[ZV^n\left(h_n/Z\right)\right]. 
				\eq
			According to Komlos Lemma, there exist $\tilde h_n\in\textrm{conv}(h_n,h_{n+1},...)$ converging $\mep$-a.s. to some $h$ which belongs to $\dmc(y)$ as the latter set, according to Proposition 3.1 in \cite{kramkov}, is closed under convergence in probability. Moreover, as $h\to zJ(h/z)$ is convex (in \cite{wu}, it is verified that also $(y,z)\to zV(y/z)$ is convex, this seems not needed here though) and $V^n\le V^m$, $m\ge n$, it follows that			
				\bqq
					\inf_{m\ge n} \e\big[ZV^m\big(\tilde h_m/Z\big)\big]
					&\le& \e\big[ZV^n\big(\tilde h_n/Z\big)\big]\\
					&\le& \sum_{m\ge n}\lambda^m \e\big[ZV^n\big(h_m/Z\big)\big]
					\;\le\; \sup_{m\ge n} \e\left[ZV^m\left(h_m/Z\right)\right].
				\eqq
			Hence, it holds that
				\bq
					\lim_{n\to\infty}v_Z^n(y)=\limsup_{n\to\infty}\e\left[ZV^n\left(h_n/Z\right)\right]
					\ge \liminf_{n\to\infty}\e\big[ZV^n\big(\tilde h_n/Z\big)\big].
				\eq
			Consequently, if it can be shown that the set $\big\{ZV_n^-(\tilde h_n/Z):n\in\mathbb{N}\big\}$ is UI, then use of Fatou's Lemma yields
				\bq
					\lim_{n\to\infty}v_Z^n(y)\ge  \liminf_{n\to\infty}\e[ZV^n(\tilde h_n/Z)]\ge \e[ZV^n(h/Z)]\ge v_Z(y), 
				\eq 
			where the last inequality follows as $h\in\dmc(y)$. Since, $v_Z^n(y)\le v_Z(y)$, this concludes the proof of the first conjugacy relation. 
			
			It remains to establish the uniform integrability $\big\{ZV_n^-(\tilde h_n/Z):n\in\mathbb{N}\big\}$. To this end, note that for $I(y)\le n$, it holds that $V_n(y)=V(y)$. As $V^-_n$ is increasing in $y$ and decreasing in $n$, it thus follows that
				\bq
					ZV_n^-(\tilde h_n/Z)\le ZV^-(\tilde h_n/Z)+ZV^-_1(U'(1)).
				\eq
			Next, we note that $Z$ is integrable. According to Lemma 3.6 in \cite{wu}, for a set $\q\subset\{\meq\ll\mep\}$ which is UI, it holds that the set $\{ZV^-(h/Z):h\in\dmc(y), Z\in\q\}$ is UI. Hence, the uniform integrability of $\{ZV^-(h/Z):h\in\dmc(y)\}$ follows as a special case thereof. Hence, $\big\{ZV_n^-(\tilde h_n/Z):n\in\mathbb{N}\big\}$ is UI which completes the proof of the first conjugacy. 
			
			The reverse conjugacy follows directly from the first one. Indeed, due to assumption $u_Z(x_0)<\infty$ for some $x_0>0$. Hence, it is finite for all $x>0$ and, furthermore, it is concave. Consequently, the reverse conjugacy follows from Theorem 12.2 in \cite{rock70}.

		Next, we turn to the proof of statement ii). Due to the conjugacy relations in i), the assumption that $v_Z(y)<\infty$ is equivalent to (cf. Note 1 in \cite{kramkov03})
		\bqn
			\lim_{x\nearrow\infty}u_z(x)/x=0.\label{twill}
		\eqn
		Hence, the $\meq$-uniform integrability of $(U^+(g_n))$, $n\ge 1$, can be established as in Lemma 1 in \cite{kramkov03}. For completeness, we highlight the main steps. To this end, we assume contrary to the claim that the sequence is not $\meq$- uniformly integrable. Then, (passing if necessary to a subsequence), one can find $\alpha>0$ and a disjoint sequence $(A^n)_{n\ge 1}$ of $(\Omega,\f)$ such that, for $n\ge 1$,
		\bq
			\e^\meq[U^+(g^n)\ind_{A^n}]\ge\alpha.
		\eq
	Define the sequence of random variables $(\tilde g^n)_{n\ge 1}$ by
		\bq
			\tilde g^n=x_0+\sum_{k=1}^ng^n\ind_{A^n},
		\eq
	where $x_0:=\inf\{x>0:U(x)\ge 0\}$. For any $h\in\dmc(1)$, it then holds that under $\mep$
		\bq
			\e^\mep[\tilde g^nh]\le x_0+\sum_{k=1}^n\e^\mep[\tilde g^nh]\le x_0+nx.
		\eq
	Hence, $\tilde g^n\in\cmc(x_0+nx)$. On the other hand, it holds that under $\meq$
		\bq
			\e^\meq[U(\tilde g^n)]\ge \sum_{k=1}^n\e^\meq[U^+(\tilde g^n)\ind_{A^n}]\ge \alpha n,
		\eq
	and therefore,
		\bq
			\limsup_{x\to\infty}\frac{u_Z(x)}{x}\ge \limsup_{x\to\infty}\frac{\e^\meq[U(\tilde g^n)]}{x_0+nx}\ge \limsup_{x\to\infty}\frac{\alpha n}{x_0+nx}=\alpha>0.
		\eq
	This contradicts the assumptions (cf. \eqref{twill}) which concludes the proof. 
\end{proof}

\bibliography{refs}        
\bibliographystyle{plain}  

\end{document}